\documentclass[11pt]{article}

\usepackage[a4paper]{geometry}
\usepackage{graphicx,subfigure, url}
\usepackage{amsfonts,amsmath,amssymb,amsthm}
\usepackage{authblk}
\usepackage{color}
\usepackage{comment}
\usepackage[]{algorithm2e}

\newtheorem{example}{Example}
\newtheorem{theorem}{Theorem}
\newtheorem{lemma}{Lemma}

\newcommand{\Var}{\text{\rm{Var}}}
\def\E{\mathbb{E}}

\def\var{\text{Var}}
\def\rhop{\rho_{\mathsf{Pearson}}}
\def\rhom{\rho_m}
\def\rhor{\rho_r}
\begin{document}
\title{On the Monotone Measure of Correlation}
%\author{Author 1, Author 2
\author{Omid Etesami and Amin Gohari
}
% The paper headers
\markboth{ }%
{ }
% make the title area
\allowdisplaybreaks
\date{}
\maketitle

\begin{abstract}
Based on the notion of maximal correlation, Kimeldorf, May and Sampson (1980)
 introduce a measure of correlation between two random variables, called the ``concordant monotone correlation" (CMC). We revisit, generalize and prove new properties of this measure of correlation. It is shown that CMC captures various types of correlation detected in measures of rank correlation
like the Kendall tau correlation. 
We show that the CMC satisfies the data processing and tensorization properties 
(that make ordinary maximal correlation applicable to problems in information theory). Furthermore, CMC is shown to be intimately related to the FKG inequality. Furthermore, a combinatorical application of CMC is given for which we do not know of another method to derive its result. Finally, we study the problem of the complexity of the computation of the CMC, which is a non-convex optimization problem with local maximas. We give a simple but exponential-time algorithm that is guaranteed to output the exact value of the generalized CMC.
\end{abstract}

\section{Introduction}

To quantify the correlation between two random variables $X$ and $Y$, various measures have been introduced. Perhaps the simplest one is the  Pearson correlation coefficient
$$\rhop(X;Y)=\frac{\text{Cov}(X,Y)}{\sqrt{\var{(X)}\var{(Y)}}}$$ that measures
only the linear dependence between $X$ and $Y$. On the other hand, the Hirschfeld-Gebelein-R\'enyi  \emph{maximal correlation} $\rho_m(X,Y)$ is a measure of correlation that captures more general types of dependencies \cite{Hirschfeld, Gebelein, Renyi}. It is obtained by taking the maximum of the Pearson correlation coefficient of arbitrary functions of $X$ and $Y$:
\begin{equation*}\rhom(X;Y) :=\max_{f(\cdot), g(\cdot)}\frac{\text{Cov}(f(X),g(Y))}{\sqrt{\var{(f(X))}\var{(g(Y))}}}.\end{equation*}
Maximal correlation has found
interesting applications in information theory and statistics (e.g. \cite{KA12,KangUlukus}) thanks to its data processing and tensorization properties.

Kimeldorf, May and Sampson introduce a measure of correlation called the \emph{concordant monotone correlation} (CMC) defined as \cite{KimeldorfMaySampson}
\begin{equation}\label{rhor-equation} \rhor(X;Y) :=\max\frac{\text{Cov}(f(X),g(Y))}{\sqrt{\var{(f(X))}\var{(g(Y))}}},\end{equation}
where the maximum is over all \emph{monotonically non-decreasing} functions $f$ and $g$.
Not being aware of the previous works on dependent measures of correlations, we rediscovered this measure of correlation  in the shorter version of this paper in \cite{ourletter} and called it ``maximal rank correlation" due to a new property of this measure that we prove in this paper.\footnote{Thanks to Mr. Shahab Asoodeh for drawing our attention to the work of Kimeldorf, et al.} 

Clearly, CMC is at least equal to the Pearson coefficient and at most equal to the maximal correlation. 
We observe that one could have defined CMC using the equivalent formula
$ \rhor(X;Y) :=\max \E[f(X)g(Y)]$
where the maximum is over all montonically non-decreasing functions $f$ and $g$ such that $\E[f(X)] = \E[g(Y)] = 0$ and $\E[f^2(X)] = \E[g^2(Y)] = 1.$
Notice that $-1 \le \rhor(X;Y) \le 1$. Throughout this paper we assume discrete random variables $X$ and $Y$, even though CMC can be defined for continuous variables as well.

The following properties motivate CMC:

1. The correlation $\rhor(X,Y)$ is the same as $\rhor(f(X),g(Y))$ for all strictly increasing functions $f$ and $g$. Contrast this with the Pearson coefficient where $\rhop(X;Y)$ can be drastically different from $\rhop(\log X; \log Y)$, i.e. the Pearson correlation is different if we represent the points on the logarithmic scale. \color{black}
(This property is similar to a difference between the median and the expected value:
For median we have $f[\text{median}(X)]=\text{median}(f(X))$ for all strictly increasing functions $f$,
whereas $f(\E[X]) = \E[f(X)]$ only for linear functions $f$.)\color{black}

This property also explains why CMC is a 
\emph{rank} correlation.\footnote{In the case of continuous variables, it is a measure of correlation between two copulas.} Suppose we have a class of students who have taken two courses with two different professors. We would like to measure the correlation between the performance of students in the two classes, but the professors have different grading and exam practices. If we calculate the Pearson correlation (between the two grades of a random student), it is sensitive to the actual grading practices of the two professors.  It might be that the grades in one class have been normalized to be Gaussian with certain mean and variance (with a non-linear mapping), whereas in the other class they have not been normalized. The advantage of CMC is that it is not sensitive to this normalization, and only cares about the ranking of the students in the two courses.

2. Suppose $X$ and $Y$ have a continuous joint distribution,
but we do not know the distribution and have only access to sample points from the joint distribution.
If we look at the uniform distribution on the empirical sample points (instead of the real joint distribution), 
the maximal correlation is trivial; in fact, the maximal correlation is 1 since $x$ coordinates and $y$ coordinates
 of all sample points are distinct almost surely (in the empirical distribution, rv's $X$ and $Y$ will be functions of each other). 
However, the maximal rank correlation for the empirical distribution is non-trivial.  

3. It is shown in \cite{monotone-dep} that one can find a sequence $(U_n, V_n)$ whose joint distribution converges to that of $(U,V)$, while $\rho_m(U_n;V_n)=1$ for all $n$, but $\rho_m(U;V)=0$. On the other hand, from $\rho_r(U_n;V_n)=1$ for all $n$, one can conclude that $\rho_r(U;V)=1$.

4. For some random variables, like jointly Gaussian variables with positive correlation (and monotonic functions of such random variables),
CMC equals maximal correlation.
In other words, it detects all dependencies detectable by maximal  correlation.

\vspace{0.2cm}
We refer the author to \cite{Chhfty2, Chhfty, Schriever} and \cite[Ex. 4.2.3]{Schriever2} for study of the cases where $\rho_r(X;Y)=1$ or $\rho_r(X;Y)=\rho_m(X;Y)$. Computation of CMC has been studied in \cite{KimeldorfMaySampson} where it is shown that computing CMC reduces to a non-convex optimization problem with local maximas. Kimeldorf, May and Sampson had proposed an iterative algorithm for computing CMC, which may trap into local maximas. This problem is further discussed in \cite{Chhfty2, Chhfty, Sampson92}.

Our other contributions in this paper can be summarized as follows: 

(a) We define the generalization of CMC for partial orders (instead of total orders). Then, we provide a connection between CMC and the FKG inequality.

(b) As we will see, CMC is greater than or equal to previously known rank correlations such as the Spearman rank correlation \cite{Spearman} and Kendall's tau rank correlation \cite{Kendall}. While the fact that CMC is greater than or equal to the Spearman rank correlation \cite{Spearman} is known (e.g.~see \cite{Sampson92}, \cite[pp. 153-154]{book}, as far as we know its relation with Kendall's tau correlation is not known before. See \cite[Sec. I]{Farnoud} for some applications of rank correlation measures.
This property also explains why CMC may be called the \emph{maximal} rank correlation.

(c) Similar to maximal correlation, the CMC satisfies the data processing and tensorization properties. We use the tensorization property to prove property (b) above. We also show a combinatorical application in Example \ref{example3}.

(d) We give an algorithm for computing the CMC  (and its generalization).  Our algorithm is guaranteed to compute the exact value of the CMC (does not trap in local maximas).

\begin{figure*}
\begin{align}
\E[fg]&= \E_{X_1Y_1} \E_{X_2Y_2|X_1Y_1}[fg]\nonumber
\\&\leq \E_{X_1Y_1}\left[ \E_{X_2|X_1Y_1}[f]\cdot \E_{Y_2|X_1Y_1}[g] +\alpha\sqrt{\Var_{X_2|X_1Y_1}[f]\cdot \Var_{Y_2|X_1Y_1}[g]}   \,    \right] \label{(i)}\\
&= \E_{X_1Y_1}\Big[ \E_{X_2|X_1}[f]\cdot\E_{Y_2|Y_1}[g]\Big] +\alpha\E_{X_1Y_1}\Big[\sqrt{\Var_{X_2|X_1}[f]\cdot\Var_{Y_2|Y_1}[g]}\,\Big]  \nonumber     \\
&\leq  \E_{X_1}\E_{X_2|X_1}[f]\cdot\E_{Y_1}\E_{Y_2|Y_1}[g] + \alpha\sqrt{\Var_{X_1}\E_{X_2|X_1}[f]\cdot\Var_{Y_1}\E_{Y_2|Y_1}[g]} \nonumber
\\&\qquad+\alpha\E_{X_1Y_1}\Big[\sqrt{\Var_{X_2|X_1}[f]\cdot\Var_{Y_2|Y_1}[g]}\,\Big] \label{(ii)} \\
&\leq \E_{X_1}\E_{X_2|X_1}[f] \cdot \E_{Y_1}\E_{Y_2|Y_1}[g] + \alpha\sqrt{\Var_{X_1}\E_{X_2|X_1}[f]\cdot \Var_{Y_1}\E_{Y_2|Y_1}[g]}\nonumber
\\&\qquad
+\alpha\sqrt{\E_{X_1}\Var_{X_2|X_1}[f]\cdot \E_{Y_1}\Var_{Y_2|Y_1}[g]} \label{(iii)} \\
&\leq \E_{X_1X_2}[f]\cdot \E_{Y_1Y_2}[g]\nonumber \\&\qquad+\alpha      
\sqrt{\left(\Var_{X_1}\E_{X_2|X_1}[f]+\E_{X_1}\Var_{X_2|X_1}[f] \right) \left(\Var_{Y_1}\E_{Y_2|Y_1}[g]+ \E_{Y_1}\Var_{Y_2|Y_1}[g]\right)}          \label{(iv)}        \\
&=  \E_{X_1X_2}[f]\cdot \E_{Y_1Y_2}[g]+\alpha\sqrt{\Var_{X_1X_2}[f]\Var_{Y_1Y_2}[g]}.
\label{(v)}\\\nonumber
\\ 
\hline \nonumber
\end{align}
\vspace{-1.2cm}
\end{figure*}

\section{Properties of CMC}
\subsection{CMC for partial orders}
Consider our example of grading students in two courses. For each of the two professors, 
it may be easy to compare some students with each other but difficult to fairly give a full ranking of all students.
So we may be interested in finding the amount of consistency (correlation) between two \emph{partial} rankings provided by the two professors.
This motivates the definition below.

Let $(\mathcal{X}, \preceq_\mathcal{X})$ and $(\mathcal{Y}, \preceq_\mathcal{Y})$ be two partially ordered sets. We say that the real-valued function $f$ is monotonic on $\mathcal{X}$ if $f(x)\leq f(x')$ when $x \preceq_\mathcal{X} x'$. Take some random variables $X$ and $Y$ with some $p(x,y)$ over $\mathcal{X}\times \mathcal{Y}$. The (generalized) CMC $\rhor(X;Y)$ is defined as in Eqn.~\eqref{rhor-equation}, except that the maximum is taken over functions $f$ and $g$ monotonic with respect to  the partial orders $(\mathcal{X}, \preceq_\mathcal{X})$ and $(\mathcal{Y}, \preceq_\mathcal{Y})$. For example, when no two members of $\mathcal{X}$ and $\mathcal{Y}$ are comparable, we have $\rhor(X;Y)=\rhom(X;Y)$. On the other hand, when $\mathcal{X}$ and $\mathcal{Y}$ are subsets of the reals with the natural ordering of real numbers, we get the earlier definition of CMC.

\begin{example}\label{example1rh}
Suppose that $\mathcal{X} = \{0,1\}^n$ with the partial order $(x_1, \ldots, x_n) \preceq (x'_1, \ldots, x'_n)$ iff $x_i \le x'_i$ for $i = 1, \ldots, n$. Suppose that $\mathcal{Y} = \mathcal{X}$ but with the \emph{reverse} partial order.
Let $X = (X_1, \ldots, X_n) \in \mathcal{X}$ be a random variable such that $X_1, \ldots, X_n$ are independent. Let $Y = X$.
Then according to the FKG inequality \cite{FKG}, we have $\emph{Cov}(f, g) \le 0$ when $f$ is monotonic on $X$ and $g$ is montonic on $Y$. In other words, $\rhor(X, Y) \le 0$.
\end{example}

The FKG inequality is a  nontrivial inequality; for example it states that any two monotone properties have positive correlation 
on an  Erd\H{o}s-R\'{e}nyi random graph.
For further applications, see \cite[Chapter 6]{AS}. 
The above example shows that FKG inequality is intimately connected to CMC. The inequality $\rhor(X, Y) \le 0$ given in Example \ref{example1rh} can be alternatively obtained using the tensorization property of CMC given later. See also Example \ref{example3}.

\subsection{Data processing and tensorization property}

Maximal correlation satisfies the data processing and tensorization properties:
\begin{align*}
 \rhom&(X';Y')\leq \rhom(X;Y) ~\text{\ if \ }~ p(x',y'|x,y)=p(x'|x)p(y'|y),\\
\rhom&(X,X';Y,Y')= \max(\rhom(X;Y), \rhom(X';Y')) ~\text{\ if \ }~ p(x',y',x,y)=p(x',y')p(x,y).
\end{align*}

A trivial variant of the data processing inequality holds for CMC:
When $\mathcal{X}, \mathcal{Y}, \mathcal{X'},$ and $\mathcal{Y'}$ are partially ordered sets,
and 
 $f: \mathcal{X} \to \mathcal{X'}$ and $g: \mathcal{Y} \to \mathcal{Y'}$ are order-preserving functions,
then $\rhor(X;Y) \ge \rhor(f(X);g(Y))$. Equality holds if $f$ and $g$ are invertible increasing functions.

To state the variant of the tensorization property for CMC, 
recall that the \emph{product order} of two partially ordered sets $(\mathcal{X}, \preceq_\mathcal{X})$ and $(\mathcal{X}, \preceq_{X'})$
is the set $\mathcal{X} \times \mathcal{X'}$ with the partial order $(x_1,x'_1) \le (x_2, x'_2)$ iff $x_1 \preceq_\mathcal{X} x_2$ and $x'_1 \preceq_\mathcal{X'} x'_2$.

First, we observe in the following example that the straightforward form of the tensorization property is not correct.
%Next, we prove that a variant of the transorization property holds anyway.

\begin{example}\label{example2n} Let $\mathcal{X}=\{0,1\}$ with the order $0\prec_\mathcal{X} 1$, and $\mathcal{Y}=\{1,0\}$ with the order $1\prec_\mathcal{Y} 0$. Let $X$ take the uniform distribution on $\mathcal{X}$, and let $Y=X$. Then $f(x) = (-1)^{1-x}$ and $g(y)=(-1)^y$ are the only possible functions with zero expected value and variance one. Hence $\rhor(X;Y)=\E[f(X) g(Y)] = -1$. On the other hand, if  $(X_1, Y_1)$ and $(X_2, Y_2)$ are i.i.d.\ copies of $(X, Y)$, then $\rhor(X_1,X_2; Y_1,Y_2) \ge 0$ with respect to the product orders. To see this, observe that
$f(x_1, x_2)=(-1)^{1-x_1}, g(y_1, y_2) = (-1)^{y_2}$ have zero mean and variance one, but $\E[f g] = 0$.
\end{example}

To state the valid variant of the tensorization property,
let $\rhor^+(X;Y)=\max(\rhor(X;Y),0)$.
\begin{theorem}\label{thm:tensoriz}
Let $X_1$, $X_2$, $Y_1$ and $Y_2$ be random variables taking values in partially ordered sets $(\mathcal{X}_1, \preceq_{\mathcal{X}_1})$, $(\mathcal{X}_2,\preceq_{\mathcal{X}_2})$, $(\mathcal{Y}_1, \preceq_{\mathcal{Y}_1})$ and $(\mathcal{Y}_2, \preceq_{\mathcal{Y}_2})$ respectively.
Then
$$\rhor^+(X_1,X_2;Y_1,Y_2)= \max(\rhor^+(X_1;Y_1), \rhor^+(X_2;Y_2))$$  if $p(x_1,y_1,x_2,y_2)=p(x_1,y_1)p(x_2,y_2)$,
where the CMCs are computed with respect to the product orders $\mathcal{X}_1 \times \mathcal{X}_2$ and $\mathcal{Y}_1 \times \mathcal{Y}_2$.
\end{theorem}
\begin{proof}
The proof follows Kumar's proof \cite{Kumar10b} of tensorization of maximal correlation. Clearly, $$\rhor^+(X_1,X_2;Y_1,Y_2)\geq  \max(\rhor^+(X_1;Y_1), \rhor^+(X_2;Y_2))$$ since any monotonic function of $x_1$ is also a monotonic function of $(x_1, x_2)$. For the other direction, let $\alpha=\max(\rhor^+(X_1;Y_1), \rhor^+(X_2;Y_2))$. It suffices to show that 
 for monotonic functions $f(x_1, x_2)$ and $g(y_1, y_2)$,
$$\E[fg]\leq \E_{X_1X_2}[f]\E_{Y_1Y_2}[g]+\alpha\sqrt{\Var_{X_1X_2}[f]\Var_{Y_1Y_2}[g]}.$$

For this we use the derivation given on the top of this page. Here in \eqref{(i)} we use the definition of CMC for the conditional distribution $p_{X_2Y_2|X_1=x_1,Y_1=y_1}$ for all $(x_1, y_1)$ 
(and use the fact that $x_2 \mapsto f(x_1, x_2)$ is monotonic for every fixed value of $x_1$), and then we  take average over all those inequalities  for all $(x_1, y_1)$ . In \eqref{(ii)} we use  the definition of CMC for distribution $p_{X_1Y_1}$ applied to functions $x_1\mapsto\E_{X_2|X_1=x_1}[f]$ and $y_1\mapsto\E_{Y_2|Y_1=y_1}[g]$. These functions are monotonic because $\E_{X_2|X_1=x_1}[f]=\sum_{x_2}p(x_2)f(x_1,x_2)$ is monotonic.
In \eqref{(iii)} and \eqref{(iv)} we use the Cauchy-Schwarz inequality, and in \eqref{(v)} we use the law of total variance. Notice that in steps \eqref{(iii)} and \eqref{(iv)}, we use the fact that $\alpha\geq 0$.
\end{proof}

\begin{example}\label{example3} Consider the following application: let $f$ and $g$ be two arbitrary balanced increasing boolean functions on $\mathbb{Z}_k^n=\{0,1,\dots, k-1\}^n$. Then $f(X^n)$ at uniformly random $X^n \in \mathbb{Z}_k^n$ is with at least some positive probability (independent of n)
different from the value of $g(Y^n)$ at $Y^n \in \mathbb{Z}_k^n$ 
when  $Y_i = (X_i + 1 \text{ mod } k)$ for all $i$.
We do not know of another method to derive this result e.g.\ using traditional correlation inequalities like the FKG inequality, as the FKG only allows for a \underline{single} partial order for both of the two increasing functions.\end{example}

\subsection{Relation to previous rank correlations}
In this section we compare CMC with two known measures of rank correlation, namely Kendall tau correlation \cite{Kendall} and Spearman correlation coefficient \cite{Spearman}. We begin by providing a general framework that can illustrate this result. Herein, we consider CMC for two real valued random variables $X$ and $Y$ (and not the generalized CMC for partial orders).

A rank correlation between $X$ and $Y$ attempts to capture the question that how much an increase in the value for $X$ is correlated with an increase in the value for $Y$. For instance, suppose we randomly draw $(X,Y)$ and get numbers $(x_1, y_1)=(5,9)$. Then if we draw $(X,Y)$ again and get $(x_2, y_2)=(7, 10)$, we can see that the value of $x_2$ is bigger than $x_1$, and similarly $y_2$ is greater than $y_1$. This is consistent with $X$ and $Y$ providing the same rankings. To define a rank correlation measure, having two values of $x_1$ and $x_2$, we need a measure that assigns a relative rank between $x_1$ and $x_2$. Let us use a function $f(x_1, x_2)$ to measure how much $x_1$ is larger  than $x_2$. We make the following assumptions about $f$:

\begin{enumerate}
\item$f(x_1,x_2)\leq f(x'_1, x_2)$ if $x_1<x'_1$, \emph{i.e.,} $f(x, x_2)$ is monotonically non-decreasing in $x$ for all $x_2$;
\item$f(x_1,x_2)\leq f(x_1, x'_2)$ if $x_2>x'_2$, \emph{i.e.,} $f(x_1, x)$ is monotonically non-increasing in $x$ for all $x_1$.
\end{enumerate}

Similarly $g(y_1, y_2)$ compares the rank of $y_1$ with $y_2$ and satisfies the above properties. Assuming that $(X_1, Y_1)$ and $(X_2, Y_2)$ are i.i.d.\ distributed according to a given $p(x,y)$, a measure of rank correlation between $X$ and $Y$ may be defined as $\rhop(f(X_1, X_2);g(Y_1, Y_2))$.

Observe that setting $f(x_1,x_2)=\mathsf{sign}(x_1-x_2)$ (with $\mathsf{sign}(0)=0$) and $g(y_1,y_2) = \mathsf{sign}(y_1 - y_2)$ gives us the Kendall tau correlation. Setting $f(x_1,x_2)=\Phi_X^{-1}(x_1)-\Phi_X^{-1}(x_2)$ and $g(y_1, y_2) = \Phi_Y^{-1}(y_1)-\Phi_Y^{-1}(y_2)$ gives us the Spearman correlation coefficient, if $\Phi_X(x)$ is the CDF of $X$. Therefore, we get known measures of rank correlation with specific choices for $f$ and $g$. The following theorem says that the CMC (if non-negative) is greater than or equal to Kendall tau and Spearman correlation coefficients.
\begin{theorem}
Let $(X_1, Y_1)$ and $(X_2, Y_2)$ be two i.i.d.\ repetitions of the pair of random variables $(X,Y)$.
Then
$$\rhor^{+}(X;Y) = \max(0, \max_{f, g} \rhop(f(X_1, X_2);g(Y_1, Y_2)) )$$
where the maximum is over 
\emph{all} non-zero functions $f$ and $g$ satisfying properties 1 and 2 above.
\end{theorem}
\begin{proof}
Let us define the order on $\mathcal{X}_1$ and $\mathcal{Y}_1$ to be the natural order on reals, but the order on $\mathcal{X}_2$ and $\mathcal{Y}_2$ be reversed. Then, the two properties of $f$ and $g$ will be the monotonicity properties for the product sets $\mathcal{X}_1\times \mathcal{X}_2$ and $\mathcal{Y}_1\times \mathcal{Y}_2$. By Theorem~\ref{thm:tensoriz}, we get that the maximum over \emph{all} valid functions $f$ and $g$ satisfying the two properties, if non-negative, is equal to $\max(\rhor^+(X_1;Y_1), \rhor^+(X_2; Y_2))$. However, $\rhor^+(X_1;Y_1)=\rhor^+(X_2; Y_2)$ since reversing the order does not change the CMC: a function $f(x)$ is monotonic with respect to an order, if and only if $-f(x)$ is monotonic with respect to the reverse order; similarly for $g(y)$. And the Pearson correlation coefficient of $f(X)$ and $g(Y)$ is the same as that of $-f(X)$ and $-g(Y)$.
This completes the proof.
\end{proof}

\section{Computation of CMC}
Computation of maximal correlation and the optimizers $f(\cdot)$ and $g(\cdot)$ are easy using the connection between this problem and the second singular value of a certain matrix \cite{Witsenhausen75}.
However, for computing CMC we do not yet know of polynomial time algorithms when $\mathcal{X}$ and $\mathcal{Y}$ are big.
Nevertheless, we show that CMC is computable
by giving an exponential time algorithm for it.
The computational complexity of CMC is an interesting open problem.

\begin{theorem}
Suppose $(\mathcal{X}, \preceq_\mathcal{X})$ and $(\mathcal{Y}, \preceq_\mathcal{Y})$ are two given finite partially ordered sets.
Furthermore, suppose a joint pmf on $\mathcal{X} \times \mathcal{Y}$ is given.
We can compute $\rho_r(X;Y)$ in time polynomial in $|\mathcal{X} \times \mathcal{Y}|$ but exponential in the total number of 
 inequality relations in the two partially ordered sets. 
\end{theorem}
In the above theorem, when we count the total number of inequality relations, 
we do not count the trivial inequalities $x \preceq x'$ where $x = x'$; that is, we only count strict inequalities.

\begin{proof}
Let $R_\mathcal{X}$ be the set all pairs $(x, x')$ such that $x \ne x'$ but $x \preceq_\mathcal{X} x'$; similarly define $R_\mathcal{Y}$.
Algorithm~\ref{alg} computes $\rho_r(X;Y)$ as desired.
The maximization of $\text{Cov}(f,g)$ in the \textbf{for} loop
can be done in polynomial time, since 
it is the problem of computing ordinary maximal correlation
         where $x$ and $x'$ are merged together for $(x, x') \in S_\mathcal{X}$
and 
$y$ and $y'$ are similarly merged when $(y, y') \in S_\mathcal{Y}$. When two symbols $x$ and $x'$ are merged, the resulting merged symbol has a probability of occurance which is the sum of the probabilities assigned to $x$ and $x'$.

\begin{algorithm}\label{alg}
$\rho_r \leftarrow -\infty $\;
let $R_\mathcal{X}$ be the set all pairs $(x, x')$ such that
$x \ne x'$ but $x \preceq_\mathcal{X} x'$\;
similarly define $R_\mathcal{Y}$\;
\For{
all $S_\mathcal{X} \subseteq R_\mathcal{X}$ and 
$S_\mathcal{Y} \subseteq R_\mathcal{Y}$
}{
   maximize $\text{Cov}(f,g)$ subject to \\ 
   \hspace{2mm} $\Var[f] = \Var[g] = 1$,
   $f(x) = f(x')$ for $(x, x') \in S_\mathcal{X}$ \\
   \hspace{2mm} and 
   $g(y) = g(y')$ for $(y, y') \in S_\mathcal{Y}$\;
   let $\tilde f$ and $\tilde g$ maximize $\text{Cov}(f,g)$\;
   \If{either of $(\tilde f, \tilde g)$ or $(-\tilde f, -\tilde g)$ are monotone
  \\ \hspace{2mm} with respect to $\preceq_\mathcal{X}$
           and $\preceq_\mathcal{Y}$}{
		 $\rho_r \leftarrow \max(\rho_r, \text{Cov}(\tilde f, \tilde g))$}
}
 \caption{Computing generalized CMC}
\end{algorithm}

To prove that the algorithm works, first observe that $\rho_r$ in the algorithm is updated only when monotone functions are found. Therefore, the output of the algorithm is less than or equal to the real CMC. To show the other direction, let $f^*$ and $g^*$ be the optimum monotone functions through which the CMC is obtained. Furthermore, among all such pairs $f^*$ and $g^*$, choose one that has the maximum number of equalities $f^*(x)=f^*(x')$ or $g^*(y)=g^*(y')$ for the comparable pairs $(x, x')$ and $(y, y')$. Then, let us consider the following relaxation of the problem: we drop the relation $x \preceq_\mathcal{X} x'$ if $f^*(x)<f^*(x')$ and similarly for $g^*(\cdot)$. 
We claim that given any maximizer $\tilde f$ and $\tilde g$
of the relaxed problem, either  $(\tilde f, \tilde g)$ or $(-\tilde f, -\tilde g)$ are monotone with respect to $\preceq_\mathcal{X}$
           and $\preceq_\mathcal{Y}$. This claim proves the correctness of the algorithm.

This claim follows from the following statement on maximal correlation. 
\begin{lemma} \label{lemma1a} Let $p(x,y)$ be an arbitrary distribution. Let
$\Gamma=\{(f,g): \Var[f(X)] = \Var[g(Y)] = 1\},$
and
$\Gamma^*=\{(f,g)\in \Gamma: f(x), g(y) \text{ maximize }\emph{Cov}(f, g)\}.$
Then for any arbitrary $(f_0, g_0)\in\Gamma$ and $(\tilde f, \tilde g)\in\Gamma^*$,
 there exists a path $(f_\alpha, g_\alpha)\in\Gamma$ for $\alpha\in[0,1]$ with the following properties:\begin{itemize}
\item $f_\alpha(x)$ and $g_\alpha(y)$ are continuous in $\alpha$ for all $x,y$, 
\item $(f_1, g_1)$ is either equal to $(\tilde f, \tilde g)$ or $(-\tilde f, -\tilde g)$,
\item $\emph{Cov}(f_\alpha, g_\alpha)$ is non-decreasing in $\alpha$.
\end{itemize}
\end{lemma}
We use the lemma in the following form. 
Let $(f_0, g_0)$ be the reduced form of $(f^*, g^*)$, i.e.\ 
$f_0$ and $g_0$ are the same as $f^*$ and $g^*$ except that they are 
functions of the merged symbols instead of the original symbols.
Similarly consider the reduced form of $(\tilde f, \tilde g)$. By definition, $(f_0, g_0)$ are monotone and are in $\Gamma$ for the reduced form.
Similarly, $(f_1, g_1)$ is in $\Gamma^*$.
Suppose that one of $f_1$ or $g_1$ is not monotone.
Consider the continuous path that the lemma ensures exists.
Let $\alpha^*$ be the supremum $\alpha$ such that 
$(f_\alpha, g_\alpha)$ are monotone.
Then there exists $(x,x')$ in the reduced form such that 
$x \prec x'$ but
$f_{\alpha^*}(x) = f_{\alpha^*}(x')$.
Thus,
$(f_{\alpha^*}, g_{\alpha^*})$ satisfies more equalities than $(f_0, g_0)$ (i.e.\ the reduced form of $(f^*, g^*)$),
while $\text{Cov}(f_{\alpha^*}, g_{\alpha^*}) \ge 
\text{Cov}(f_0, g_0)$ which is a contradiction.
\end{proof}
\begin{proof}[Proof of Lemma~\ref{lemma1a}] Take arbitrary zero mean functions $f$ and $g$ satisfying $\mathbb{E}[f^2]=\mathbb{E}[g^2]=1$. As shown in  \cite{Witsenhausen75}, we can write
 \begin{align*}\mathbb{E}[fg]&=\sum_{x,y}p(x,y)f(x)g(y)
\\&=\sum_{x,y}\frac{ p(x,y)}{\sqrt{p(x)p(y)}}\cdot\sqrt{p(x)}f(x)\cdot \sqrt{p(y)}g(y)
\\&=v_f^T \tilde P v_g\end{align*}
where $\tilde P$ is a $|\mathcal{X}| \times |\mathcal{Y}|$ matrix whose rows are indexed by $\mathcal{X}$ and whose columns are indexed by $\mathcal{Y}$; this matrix is defined by $\tilde{P}_{xy} = p(x,y)/\sqrt{p(x)p(y)}$. The vector $v_f$ is a real column vector of size $|\mathcal{X}|$, indexed by $\mathcal{X}$ whose $x$ entry is equal to $\sqrt{p(x)}f(x)$; $v_g$ is defined similarly. The vectors  $v_f$ and $v_g$ are of norm one since 
$$v_{f}^T v_f=\sum_{x}\sqrt{p(x)}f(x)\cdot\sqrt{p(x)}f(x)=\sum_x p(x)f(x)^2=\mathbb{E}[f^2]=1,$$
and similarly for $v_g$.

Let  $\tilde{P}=\sum_{i=1}^t\lambda_iu_iw_j^T$ be the singular value decomposition of $\tilde{P}$, where $ \lambda_1\geq \lambda_2\geq ... \geq \lambda_t > 0$, and $\{u_1, \ldots, u_t\}$ 
and $\{w_1, \ldots, w_t\}$ are 
two orthonormal sets of vectors. Then it is known that the maximum singular value  $\lambda_1=1$ and the corresponding vectors  $u_1$ and $w_1$ have entries  of the form $\sqrt{p(x)}$ and $\sqrt{p(y)}$ respectively  \cite{KangUlukus, Witsenhausen75}. The constraint that $\mathbb{E}[f]=\mathbb{E}[g]=0$, implies that $v_f$ is perpendicular to $u_1$ and $v_g$ is perpendicular to $w_1$ since for instance
$$v_{f}^T u_1=\sum_{x}\sqrt{p(x)}f(x)\cdot\sqrt{p(x)}=\sum_x p(x)f(x)=\mathbb{E}[f]=0.$$
Thus, we are interested in the maximum of $v_f^T \tilde P v_g$ over all unit vectors $v_f$ and $v_g$ that are  perpendicular to vectors $u_1$ and $v_1$ (corresponding to the maximum singular value).  It is shown in \cite{Witsenhausen75}, 
 that the solution to this problem, $\rhom(X;Y)$, is the second maximum singular value $\lambda_2$, and the vectors $v_f$ and $v_g$ that obtain it, will be the vectors $u_2$ and $w_2$, associated to the second maximum singular value.

Now, let us turn to the proof of the lemma. Since adding a constant to $f$ and $g$ does not affect $\text{Cov}(f,g)$, without loss of generality, we may assume that $\mathbb{E}[f_0]=\mathbb{E}[g_0]=\mathbb{E}[f_1]=\mathbb{E}[g_1]=0$, as one can adjust the expected values in a continuous way. Furthermore, without loss of generality, we may assume that $ u_2$ and $w_2$ (corresponding to the second singular value) are the vectors associated to maximizer functions $f_1$ and $g_1$ respectively, \emph{i.e.}, $ v_{f_1}=u_2$ and $ v_{g_1}=w_2$.

Let $v_{f_0}, v_{g_0}$ be the corresponding vectors associated to functions $f_0$ and $g_0$ respectively. Observe that 
 \begin{align*}\mathbb{E}[f_0g_0]&=v_{f_0}^T \tilde P v_{g_0}
\\&=v_{f_0}^T\big(\sum_{i=1}^t\lambda_iu_iw_j^T\big) v_{g_0}
\\&=\sum_{i=1}^t\lambda_ic_i(0)d_i(0),\end{align*}
where $c_i(0)=v_{f_0}^T u_i$ and $d_i(0)=w_j^T v_{g_0}$. We have $c_1(0)=d_1(0)=0$ since $\mathbb{E}[f_0]=\mathbb{E}[g_0]=0$. Furthermore $\sum_{i}c_i(0)^2\leq \|v_{f_0}\|^2=1$, and similarly  $\sum_{i}d_i(0)^2\leq 1$ since $u_i$ and $w_i$ are orthonormal. Let $v_{f_0}=\sum_{i=1}^t c_i(0)u_i + \Delta v_{f_0}$ and $v_{g_0}=\sum_{i=1}^t d_i(0)w_i + \Delta v_{g_0}$.

To define the functions $f_{\alpha}$ and $g_{\alpha}$, we consider functions whose corresponding vectors are of the form $\sum_{i=1}^t c_i(\alpha)u_i+ \Delta v_{f_0}$ and $\sum_{i=1}^t d_i(\alpha)w_i+ \Delta v_{g_0}$ respectively. The goal is to move from the  tuples $(c_2(0), ..., c_t(0))$ and $(d_2(0), ..., d_t(0))$ towards the tuples $$(s\sqrt{c_2(0)^2+...+c_t(0)^2}, 0,...,0)$$ and $$(s\sqrt{d_2(0)^2+...+d_t(0)^2},0,...,0)$$ for $s=1$ or $s=-1$, in a continuous way, such that $\sum_{i=2}^t\lambda_ic_id_i$ is monotonically non-decreasing on the continious path. Lemma \ref{lemma2} shows that this is possible. Thus, we have moved on a continuous path, while increasing $\mathbb{E}[fg]$, towards functions whose corresponding vectors are of the form a constant times $u_2$ plus $\Delta v_{f_0}$, and a constant times $v_2$ plus $\Delta v_{g_0}$ respectively.

Assume $s=1$ (the case of $s=-1$ is similar). Since $\sum_{i=2}^tc_i(0)^2\leq 1$, and  $\sum_{i=2}^td_i(0)^2\leq 1$, we can continiously increase the coefficient of $u_2$ and $v_2$ to one, while reducing the norm of $\Delta v_{f_0}$ and  $\Delta v_{g_0}$ appropriately to keep the norm of the functions be one (we reduce the norm of $\Delta v_{f_0}$ and  $\Delta v_{g_0}$ by multiplying them in a constant less than one). Since $\lambda_2$ is non-negative, increasing the coefficients of $u_2$ and $v_2$ will not decrease $\mathbb{E}[fg]$. Since the residual vectors $ \Delta v_{f_0}$ and $ \Delta v_{g_0}$ are perpendicular to the $t$ vectors corresponding to the first $t$ singular values, reducing their norm by multiplying them by a constant does not affect the expected value $\mathbb{E}[fg]$. Thus, we can reach functions whose vectors correspond to $u_2=v_{f_1}$ and $w_2=v_{g_1}$ in a continuous way, while making sure that $\mathbb{E}[fg]$ is non-decreasing on this path.
\end{proof}

\begin{lemma}\label{lemma2}
Suppose we are given two non-zero vectors $(c_2, ..., c_t)$ and $(d_2, ..., d_t)$ and non-negative numbers $ \lambda_2\geq ... \geq \lambda_t$. Then it is possible  to continuously change the coordinates of these two vectors to reach $(s\sqrt{c_2^2+...+c_t^2}, 0, 0, ..., 0)$ and $(s\sqrt{d_2^2+...+d_t^2}, 0, 0, ..., 0)$ for $s=1$ or $s=-1$  in a way that (i) the norm of the two vectors remain unchanged during the gradual changes in the vectors (ii) $\sum_{i=2}^t\lambda_ic_id_i$ is monotonically non-decreasing on the continious path.
\end{lemma}
\begin{proof}
By induction in $t$, and by altering the coordinates 3, ..., t, we can move the two vectors to
$(c_2, s'\sqrt{c_3^2+...+c_t^2}, 0, 0, ..., 0)$ and $(d_2,s'\sqrt{d_3^2+...+d_t^2}, 0, 0, ..., 0)$ for some $s'\in\{-1,1\}$.
Then by considering coordinates 1,2 of these vectors and using the induction hypothesis for vectors of size 2, we can reach the vectors
$(s\sqrt{c_2^2+c_3^2+...+c_t^2}, 0, 0, ..., 0)$ and $(s\sqrt{d_2^2+d_3^2+...+d_t^2}, 0, 0, ..., 0)$ for some $s\in\{-1,1\}$.

To show the inducation basis for $t=2$, consider two vectors
$(c_2, c_3)= (r_1\cos(\phi_1), r_1\sin(\phi_1))$ and
$ (d_2,d_3)=(r_2\cos(\phi_2), r_2\sin(\phi_2))$. We get that
$\sum_{i=2}^3\lambda_ic_id_i=
\frac 12 r_1r_2(\lambda_2+\lambda_3)\cos(\phi_1-\phi_2)+ \frac 12 r_1r_2(\lambda_2-\lambda_3)\cos(\phi_1+\phi_2)$. Note that the coefficients of the cosine terms are non-negative. Varying angles $(\phi_1,\phi_2)$ towards either $(0,0)$ or $(\pi, \pi)$  continuously is equivalent with varying angles $\theta_1=\phi_1+\phi_2$ and $\theta_2=\phi_1-\phi_2$ towards multiples of $2\pi$ continuously. The latter can be done such that both $ \cos(\theta_1)$ and $ \cos(\theta_2)$ vary monotonically increasing: if $\theta_1\in[0,\pi]$, we decrease $\theta_1$ towards zero, and if $\theta_1\in(\pi, 2\pi)$, we increase $\theta_1$ to $2\pi$.
\end{proof}

\section{CMC, Moment generating functions  and Independence}
We saw in Example \ref{example2n} that unlike maximal correlation, CMC can be zero for dependent random variables. However, 
for nondegenerate random variables $X$ and $Y$, the equalities $\rhor(X;Y)=\rhor(-X;Y)=0$ implies independence of $X$ and $Y$  \cite{KimeldorfMaySampson}. 
When $\mathcal{X}$ and $\mathcal{Y}$ are finite subsets of the reals with the natural ordering of real numbers, this can be also observed from the following equation
\begin{align}\max(\rhor(X;Y),\rhor(-X;Y))\geq \max_{s_1, s_2\in \mathbb{R}-\{0\}}\frac{\big|M_{X,Y}(s_1, s_2)-M_X(s_1)M_Y(s_2)\big|}{\sqrt{M_X(2s_1)-M_X(s_1)^2}\sqrt{M_Y(2s_2)-M_Y(s_2)^2}},\label{eqn:total-M}\end{align}
where $M_{X,Y}(s_1,s_2)=\mathbb{E}[e^{s_1X+s_2Y}]$ is the moment generating function of $(X,Y)$. 
If $\rhor(X;Y)=\rhor(-X;Y)=0$, from the above equation we get that $M_{X,Y}(s_1, s_2)=M_X(s_1)M_Y(s_2)$ for all $s_1, s_2$. In other words, the joint moment generating function of $p(x,y)$ is the same as the  joint moment generating function of the product distribution $p(x)p(y)$. Since moment generating function uniquely specifies the distribution of finite random variables, we get that $p(x,y)=p(x)p(y)$ and thus $X$ and $Y$ are independent.

To prove equation \eqref{eqn:total-M}, observe that the function $x\mapsto \text{sgn}(s_1)e^{s_1x}$ is  increasing monotone for any real $s_1\neq 0$, where $\text{sgn}(s)\in\{+1,-1\}$ is the sign of $s$. Thus, 
$$\rhor(X;Y)\geq \max_{s_1, s_2\in \mathbb{R}-\{0\}}\frac{M_{X,Y}(s_1, s_2)-M_X(s_1)M_Y(s_2)}{\sqrt{M_X(2s_1)-M_X(s_1)^2}\sqrt{M_Y(2s_2)-M_Y(s_2)^2}}\text{sgn}(s_1)\text{sgn}(s_2),$$
Next, observe that $M_{-X}(-s_1)=M_X(s_1), M_{-X,Y}(-s_1, s_2)=M_{X,Y}(s_1, s_2)$. Thus, 
$$\rhor(-X;Y)\geq \max_{s_1, s_2\in \mathbb{R}-\{0\}}\frac{M_{X,Y}(s_1, s_2)-M_X(s_1)M_Y(s_2)}{\sqrt{M_X(2s_1)-M_X(s_1)^2}\sqrt{M_Y(2s_2)-M_Y(s_2)^2}}\text{sgn}(-s_1)\text{sgn}(s_2).$$
The above equations imply equation \eqref{eqn:total-M}.

\bibliographystyle{IEEEtran}

\end{document}